\documentclass[11pt]{article}
\usepackage{times}
\newtheorem{definition}{Definition}
\newtheorem{remark}[definition]{Remark}
\newtheorem{example}[definition]{Example}

\newtheorem{theorem}[definition]{Theorem}
\newtheorem{proposition}[definition]{Proposition}
\newtheorem{corollary}[definition]{Corollary}
\usepackage{verbatim}
\usepackage{graphicx}
\usepackage{amssymb}
\usepackage{amsmath}
\usepackage[noend]{algpseudocode}
\usepackage{enumerate}
\usepackage[title]{appendix}

\usepackage{xcolor}

\makeatletter
\def\BState{\State\hskip-\ALG@thistlm}
\makeatother

\def\RR{\hbox{\sf I\kern-.14em\hbox{R}}}

\newcommand{\sgn}{\mathrm{sgn}}

\newcommand{\mtov}{\mathrm{vec}}

\begin{document}
\global\def\refname{{\small \bf References}}
%
%

\centerline{\large \bf
       Real Forms of the Complex Neumann System: }
\centerline{\large \bf Real Roots of Polynomial $U_{\cal S}(\lambda)$
       }
\date{}
\vspace{7mm}
\centerline{\small{Tina Novak}}
\centerline{\small \it University of Ljubljana, Faculty of Mechanical Engineering, A\v{s}ker\v{c}eva 6}
\centerline{\small \it SI-1000 Ljubljana, Slovenia}
\centerline{\small \it tina.novak@fs.uni-lj.si }
\vspace{7mm}
\centerline{\small{Janez \v{Z}erovnik}}
\centerline{\small \it University of Ljubljana, Faculty of Mechanical Engineering, A\v{s}ker\v{c}eva 6}
\centerline{\small \it SI-1000 Ljubljana, Slovenia}
\centerline{\small \it Institute of Mathematics, Physics and Mechanics, Jadranska 19}
\centerline{\small \it SI-1000 Ljubljana, Slovenia}
\centerline{\small \it janez.zerovnik@fs.uni-lj.si }
\vspace{18mm}

%

%
%

\begin{abstract}
   \noindent
   The topology of Liouville sets of the real forms of the complex generic Neumann system depends indirectly on the roots of the special polynomial $U_{\cal S}(\lambda)$. For certain polynomials, the existence and positions of the real roots, according to the suitable parameters of the system, is not obvious.
   In the paper, a novel method for checking the existence and positions of the real roots of the polynomials $U_{\cal S}(\lambda)$ is given. The method and algorithm are based on searching of a positive solution of a system of linear equations. We provide a complete solution to the problem of existence of real roots for all special polynomials in case $n=2$. This is a step closer to determining the topology of the Liouville sets.
   \\



\end{abstract}

\section{Introduction}

{\color{black}
	Consider the classical {\color{black} Neumann system in Hamiltonian mechanics} that describes the motion of a particle on {\color{black} the sphere $S^n$}
	under the influence of the quadratic potential.
	{\color{black} Natural generalizations are real forms of the complex Neumann system that acquires new integrable systems that describe the motion of a particle on hyperboloids under the influence of the quadratic potential in general $(p,q)$-metrics.}
	A useful way to understand the behavior of the system is to study the  level sets,
	roughly speaking, the sets of points in the ambient space with constant {\color{black} total} energy.
	In other words, the geometry of the corresponding level sets of the systems' moment map of the Hamiltonian system is considered.
	While this is a classical problem in Hamiltonian mechanics, the general understanding of the possible topologies in general case {\color{black} for all real forms of the complex Neumann system} is surprisingly limited.
	For example,  even for  $n=2$, the sphere and some of the hyperboloids are well studied,
	but until now there was  no complete classification of all possible topologies.
	In this paper we propose a general method that makes it possible to
	classify all possible topologies for any surface (sphere or hyperboloid).
	Here, our results are limited to $n=2$ because the
	calculations are computationally demanding.
	However, the method can be in principle used for higher dimensions, if computational complexity issues are addressed properly.
	
	Here we provide a fresh look at the problem.
	In other words, the question is  regarded as
	a  combinatorial problem that is translated into a question whether there are positive solutions to certain systems of linear equations.
	Using an old method to test the  existence of a positive solution of a system of linear equations (see \cite{Dines}),
	we are able to provide complete characterization of all possible cases and  also provide the corresponding topologies.
	Thus we provide another example how combinatorial reasoning combined with linear algebra can be used to answer questions that arise in
	topology, and deeper, in theoretical mechanics.
	Our approaches opens several challenging avenues of further research. Computational methods may be implemented more efficiently, can be possibly parallelized, thus allowing application to larger $n$. One the other side, one can investigate correspondence between the results of the method and separated parts of the phase space of a system.
	
	Besides providing a complete list of possible topologies for case $n=2$, the main contribution of this work is a general method. The basic idea behind is to switch the
	role of constants and unknowns in the systems of equations.
	It has been  well understood that the
	systems behavior corresponds to a solution to certain system of equations
	after  fixing the values of certain arbitrary constants.
	Our method, that is explained in detail in this paper,
	{\color{black} varies} the constants and asks which sets of constants give rise to
	systems with nontrivial solutions.
	We employ well known tools from topology and mechanics and in particular, an algorithm for positive solutions of a system of linear equations.
	
	The rest of the paper is organized  as follows.
	After the introduction,  the  polynomial $U_{\cal S}$ is described and the problem and the idea of finding roots is explained.
	In section 3, the basics of  real forms of the complex Neumann system are outlined.
	The method of finding a positive solution of a system of linear equations is recalled in section 4.
	In section 5, the scheme of the algorithm is presented. For the example ${\cal S}=\{1,3\}$, the procedure of the algorithm is given. Because of the monotonic conversions of equations of the system, the matrix forms of a given system and a reduced system are introduced.
	In Conclusion, we comment the results of the algorithm and describe the topology of all possible Arnold-Liouville level sets for ${\cal S}=\{1,3\}$. After that, open issues are mentioned. In appendices, concrete operation of the algorithm, results of algorithm and additional analysis are given.
}

\section{{\color{black}
		Polynomial $U_{\cal S}(\lambda)$ and {\color{black} the main idea}}}

Motivated by a classical theory of Hamiltonian systems, a question arises how to determine the topology of the corresponding Arnold-Liouville level sets\footnote{Arnold-Liouville level set is the regular level set of the system's moment map}.
It was observed in \cite{Mumford}
that the  topology of the Arnold-Liouville level set of the Neumann system depends on the positions of the roots of the polynomial
\begin{equation*}
U(\lambda)=\prod_{i=1}^{n+1}(\lambda-a_i)\sum_{j=1}^{n+1}\frac{q_j^2}{\lambda-a_j}
\end{equation*}
of degree $n$ with real non-fixed coefficients where the given constants $a_1,a_2,\ldots,a_{n+1}$ are arranged as $a_1 < a_2 < \ldots < a_{n+1}$ and the equality $q_1^2 + q_2^2 + \ldots + q_{n+1}^2=1$ holds. Assuming that the roots $\lambda_1,\lambda_2,\ldots,\lambda_n$ of the polynomial $U(\lambda)$ do not coincide with the given constants $a_1,a_2,\ldots,a_{n+1}${\color{black}.}
{\color{black}Since $\sgn (U(a_{n+1-k}))=(-1)^k$, all roots $\lambda_1, \lambda_2,\ldots,\lambda_n$ are real and arranged:
	$$ a_1 < \lambda_1 < a_2 < \lambda_2 < \ldots < a_n < \lambda_n < a_{n+1}.$$}

Certain real forms of the complex Neumann system provide new real Hamiltonian systems. Thus, the topologies of the corresponding Arnold-Liouville level sets depend on the positions of the roots of 
{\color{black} the polynomials $U_{\cal S}(\lambda)$ defined as follows. }

{\color{black}
	Let ${\cal S} \subseteq \{1,2,\ldots,n+1\}$  be a nonempty subset, $\epsilon_{{\cal S},j}$ the discrete sign function
	\begin{equation}
	\epsilon_{{\cal S},j} = \left\{ \begin{array} {rl}
	1  & ; \;\; j \in {\cal S} \\
	-1 & ; \;\; j \in {\cal S}^{\mathsf c} \,,
	\end{array} \right.
	\label{eq:epsilon}
	\end{equation}
}{\color{black}
	and $A(\lambda)$ the polynomial $A(\lambda)=\prod_{i=1}^{n+1}(\lambda-a_i)$.
	Real forms of the complexification of the polynomial $U(\lambda)$
	{\color{black}
		(for detailed explanation of the complexification and the family of real forms see \cite{Novak})
	}
	are polynomials
	\begin{equation}
	U_{\cal S}(\lambda) = A(\lambda) \sum_{j=1}^{n+1} \epsilon_{{\cal S},\j} \frac{q_j^2}{\lambda-a_j},
	\label{eq:USLambda}
	\end{equation}
	where the constraint $\sum_{j=1}^{n+1} \epsilon_{{\cal S},j}q_j^2 = 1$ holds.
}
As the positions of the roots can not be easily determined for given coefficients, obtained from constants $a_1,a_2,\ldots,a_{n+1}$ and constrained variables $q_1,q_2,\ldots,q_{n+1}$, in general, only partial answers have been known (see \cite{Novak}).
{\color{black}
	In \cite{Novak}, it is proved that the polynomial $U_{\cal S}(\lambda)$
	has $n$ real roots if ${\cal S}=\{m,m+1,\ldots,k\}$ for some $m \in \{1,2,\ldots,n+1\}$ and some $k \in \{m,m+1,\ldots,n+1\}$ taking into account the constraint $\sum_{j=1}^{n+1}\epsilon_{{\cal S},j}q_{j}^{2}=1$. If the set ${\cal S}$ does not contain consecutive integers, {\color{black} it is not possible to deduce at the positions and the existence of all real roots with respect to the signs of the values $U_{\cal S}(a_j)$, $j=1,2,\ldots,n+1$. Such a "problem" is encountered even in case $n=2$:}
	\begin{example}
		\label{ex:11}
		For $n=2$ and ${\cal S}=\{1,3\}$, we have the polynomial
		\begin{eqnarray*}
			U_{\{1,3\}}(\lambda) & = & (\lambda-a_{1})(\lambda-a_{2})(\lambda-a_{3})\left(\frac{q_{1}^{2}}{\lambda-a_{1}} - \frac{q_{2}^{2}}{\lambda-a_{2}} + \frac{q_{3}^{2}}{\lambda-a_{3}}\right) = \nonumber\\
			& = & q_{1}^{2}(\lambda-a_{2})(\lambda-a_{3}) - q_{2}^{2}(\lambda-a_{1})(\lambda-a_{3}) + q_{3}^{2}(\lambda-a_{1})(\lambda-a_{2})\,. \label{eq:polynom13}
		\end{eqnarray*}
		Since $q_{1}^{2}-q_{2}^{2}+q_{3}^{2}=1$, the leading coefficient of $U_{\{1,3\}}$ is $1$. Furthermore, $\sgn(U_{\{1,3\}}(a_{1}))=\sgn(U_{\{1,3\}}(a_{2}))=\sgn(U_{\{1,3\}}(a_{3}))=1$, and we can not conclude easily about the existence or positions of real roots of the polynomial.
	\end{example}
	\begin{remark}
		{\color{black} The polynomial $U_{\{1,3\}}(\lambda)$ can be written as
			$$ U_{\{1,3\}}(\lambda) = \lambda^2 + \Big( -(a_2+a_3)q_1^2+(a_1+a_3)q_2^2-(a_1+a_2)q_3^2 \Big)\lambda + a_2a_3q_2^2-a_1a_3q_2^2+a_1a_2q_3^3 .$$
			Its discriminant is}
		$$ D = \left(-(a_2+a_3)q_{1}^{2} + (a_1+a_3)q_{2}^{2}-(a_1+a_2)q_{3}^{2}\right)^{2} - 4(a_2a_3q_{1}^{2}-a_1a_3q_{2}^{2}+a_1a_2q_{3}^{2}) \,.$$
		{\color{black}For fixed constants $a_1=1$, $a_2=2$, and $a_3=3$, it
			{\color{black}
				can be easily checked, that
			}
			the discriminant $D$ at the point $(q_1,q_2,q_3) = (1,1,1)$ is negative but at $(q_1,q_2,q_3)=(1,3,3)$ is positive.}
		Even if we find a set of points $(q_1,q_2,q_3) \in {\cal H}_{\{1,3\}}^{2}=\{(q_1,q_2,q_3)\in{\mathbb R}^3\,\big|\; q_1^2-q_2^2+q_3^2=1\}$ for which the discriminant $D$ is positive (hence real roots exist),
		{\color{black} this may not give much information about}
		the positions of real roots according to
		{\color{black} other values of}
		parameters $a_1$, $a_2$ and $a_3$.
	\end{remark}
	\begin{remark}
		In general, for a polynomial of degree $n$, the necessary condition for the existence of real roots is the positivity of the discriminant. If $n \geq 4$, the positivity of the discriminant is not a sufficient condition.
	\end{remark}
	{\color{black} For specific values $a_1,a_2,\ldots,a_{n+1}$ and for the given initial point $(q_1,q_2,\ldots,q_{n+1})_0$ of the system, the roots of the polynomial $U_{\cal S}(\lambda)$ can be determined by some of well known methods for finding the roots of polynomials (see \cite{Demmel,Mekwi}). Consequently, in such special case, if the relations between the computed roots and other $n$ suitable parameters of the system are known, the topology of the Arnold-Liouville level set follows.
		But our aim is to try to classify or partially classify the families of the polynomials real forms $U_{\cal S}(\lambda)$ according to the existence of real roots in general, i.e. for arbitrary ordered constants $a_1,a_2,\ldots,a_n$. }
	\\
	
	At this situation it makes sense to ask the following questions.
	
	{\color{black}
		1. Can we find, in general or in some special case, the set of $(q_{1},q_{2},\ldots,q_{n+1}) \in {\cal H}_{\cal S}^{n}$ for which the polynomial
		\hspace{0.4cm}$U_{\cal S}(\lambda)$ has only real roots?
		
		2. Can we determine, in general or in some special case, positions of $n$ real roots of polynomial $U_{\cal S}(\lambda)$ according
		\hspace{0.4cm}to the given parameters $a_{1},a_{2},\ldots,a_{n+1}$?
	}
	3. Can we describe, in general or in some special case regarding $n$ or ${\cal S}$, the topology of the Arnold-Liouville set? \\
	
	\noindent
	{\color{black} Since the polynomial $U_{\cal S}(\lambda)$ has very special form and for given $\lambda$, where $\lambda \ne a_{j}$, $j=1,2,\ldots,n+1$, the signs of the linear factors $(\lambda - a_j)$ follow.}
	In our approach the key idea is to replace the roles of constants and variables. Instead of assuming that $q_{1},q_{2},...,q_{n+1}$ are (given) coordinates of a point on the hyperboloid ${\cal H}_{\cal S}^{n}$ and searching for roots of polynomial $U_{\cal S}(\lambda)$, we fix values (potential roots) $\lambda_{1},\lambda_{2},\ldots,\lambda_{n}$ up to the intervals $(-\infty,a_{1})$, $(a_{i},a_{i+1})$ for $i=1,2,\ldots,n$ and $(a_{n+1},\infty)$, and ask whether a positive solution $(q_{1}^{2},q_{2}^{2},\ldots,q_{n+1}^{2})$ of the system of linear equations
	\begin{eqnarray}
	\sum_{j=1}^{n+1} \epsilon_{{\cal S},j} q_{j}^{2} & = & 1  \nonumber\\
	\sum_{j=1}^{n+1} \epsilon_{{\cal S},j} \frac{A(\lambda_{1})}{\lambda_{1}-a_{j}}q_{j}^{2} & = & 0  \nonumber\\
	\sum_{j=1}^{n+1} \epsilon_{{\cal S},j} \frac{A(\lambda_{2})}{\lambda_{2}-a_{j}}q_{j}^{2} & = & 0  \label{eq:system}\\
	&  \vdots &   \nonumber\\
	\sum_{j=1}^{n+1} \epsilon_{{\cal S},j} \frac{A(\lambda_{n})}{\lambda_{n}-a_{j}}q_{j}^{2} & = & 0 \nonumber
	\end{eqnarray}
	exists. By
	{\color{black}
		setting the
	}
	positions of real parameters (roots) $\lambda_{1}, \lambda_{2}, \ldots, \lambda_{n}$, the signs of all coefficients of system (\ref{eq:system}) are determined.
	Recall that from the definition of roots of a polynomial and the additional condition $\sum_{j=1}^{n+1}\epsilon_{{\cal S},j}q_{j}^{2}=1$ for coordinates $q_{1},q_{2},\ldots,q_{n+1}$, the following statement is obvious.
	Let $\sum_{j=1}^{n+1}\epsilon_{{\cal S},j}q_{j}^{2}=1$.
	Real values $\lambda_{1},\ldots,\lambda_{n}$ are the roots of the polynomial
	$U_{\cal S}(\lambda)$ (see (\ref{eq:USLambda}))
	for some $q_{1},q_{2},\ldots,q_{n+1}$ if and only if the system of $n+1$ non-homogeneous linear equations
	(\ref{eq:system}) admits a positive solution $(q_{1}^{2},q_{2}^{2},\ldots,q_{n+1}^{2})$.
	
	In {\color{black} this} paper, the method for checking the existence of a positive solution of a linear homogeneous system with real coefficients is recalled (from \cite{Dines}). As is seen in one of low dimension cases, elaborated in subsection \ref{subsection:alg n2}, the procedure is time-consuming. Hence the problem yells after the computer computation. {\color{black}
		In order to implement the method, we develop both
		the matrix notation of the problem for finding
		the positive solution of a linear system and the matrix notation of the algorithm
		(see Proposition~\ref{prop:reducedmatrix}).
		The algorithm for checking the existence of a positive solution of the system (\ref{eq:system}) was run in MATLAB.
		For illustration,
	}
	two examples of scripts and two main functions are given.
	By this simple method we obtain a new result towards completely determining the Arnold-Liouville level sets of the real forms of the complex Neumann system.
	For the case $n=2$, results of the algorithm for all nonempty subsets ${\cal S} \subseteq \{1,2,3\}$ are in Table \ref{Tab:1}, Table \ref{Tab:2} and Table \ref{Tab:3}. For the subsets $\{1\}$, $\{2\}$, $\{3\}$, $\{1,2\}$, $\{2,3\}$ and $\{1,2,3\}$, the results are trivial and expected.
	{\color{black} For the subset $\{1,3\}$,
		{\color{black}
			using our method we provide a complete list of possible cases with the
			corresponding topologies.
			In particular,
		}
		we find out, that for arbitrary parameters $a_1 < a_2 < a_3$ the polynomial $U_{\{1,3\}}(\lambda)$ has real roots for some $(q_{1},q_{2},q_{3}) \in {\cal H}_{\{1,3\}}^{2}$ (see Table~\ref{Tab:2}, where the middle subtable represents the results of the algorithm for the case ${\cal S}=\{1,3\}$. Number 1 in the second column means that a root in the specific interval can exists, 0 means that in the specific interval is no root).
	}
	{\color{black}
		Furthermore,  for each case where
	}
	{\color{black}
		the polynomial $U_{\{1,3\}}(\lambda)$ has real roots}, the Arnold-Liouville level sets of the Hamiltonian system $(T^{*}{\cal H}_{\{1,3\}}^{2},\omega_{\{1,3\}},H_{\{1,3\}})$ are exactly determined.

	{\color{black}
		\section{Real forms of the Complex Neumann System {\color{black} and the topology of the Arnold-Liouville level set}}
		
		{\color{black}
			In the Hamiltonian mechanics, the geometry of the Arnold-Liouville level sets provides very important information about a Hamiltonian system. It gives the topology of the set of all possible point's positions of a system in the ambient space.
			
			The Neumann system describes the motion of a particle on the sphere $S^{n}$ under the influence of a quadratic potential. For $n=2$, the system is introduced in \cite{Neumann}. Let $(q,p)=(q_{1},\ldots,q_{n+1},p_{1},\ldots,p_{n+1}) $ $ \in T^{*}{\mathbb R}^{n+1}$ be ambient coordinates and $a_{1},\ldots,a_{n+1}$ real constants. Throughout the paper, we assume
			\begin{equation*}
			a_{1} < a_{2} < \ldots < a_{n+1}\,.
			\end{equation*}
			This corresponds to the so-called generic case of the Neumann system. {\color{black} Some generalizations of the Neumann system were studied in \cite{Saksida}. In \cite{Vuk}, the integrability and the Arnold-Liouville level sets of the confluent ($a_n=a_{n+1}$) Neumann system were investigated. The general degenerate Neumann system and its reduction were researched in \cite{Dullin1}.}
			
			In the Hamiltonian formalism, the Neumann system can be written as a triple $(T^{*}\!S^{n},\omega,H)$ where $T^{*}\!S^{n}$ is the cotangent bundle of the sphere $S^{n}$, $\omega$ is the canonical $2$-form in ambient coordinates, $\omega=\sum_{j=1}^{n+1}dq_{j} \wedge dp_{j}$, and
			\begin{equation*}
			H = H(q,p) = \frac{1}{2}\left( \sum_{j=1}^{n+1}q_{j}^{2}\sum_{j=1}^{n+1}p_{j}^{2}-(\sum_{j=1}^{n+1}q_{j}p_{j})^{2}\right)+\frac{1}{2}\left(2-\sum_{j=1}^{n+1}q_{j}^{2}\right)\sum_{j=1}^{n+1}a_{j}q_{j}^{2}
			\end{equation*}
			is the Hamiltonian function. The wedge product $\wedge$ is the exterior product of forms and it is clearly explained in \cite{Arnold} (Chapter 7 and Chapter 8). The natural complexification converts real variables $(q,p)$ into complex $(Q,P) \in T^{*}{\mathbb C}^{n+1}$, and the Neumann system into the complex Neumann system $(T^{*}(S^{n})^{\mathbb C}, \omega^{\mathbb C}, H^{C})$ (see \cite{Mumford, Novak}, see \cite{Audin} for the meaning of such complex systems.
			Hence, the classical Neumann system as described above is the real form of the complex Neumann system
			{\color{black}
				with
			}
			regard to the standard antiholomorphic involutive automorphism
			\begin{eqnarray*}
				\tau: T^{*}{\mathbb C}^{n+1} & \longrightarrow   & T^{*}{\mathbb C}^{n+1} \\
				(Q_{1},\ldots,Q_{n+1},P_{1},\ldots,P_{n+1}) & \longmapsto & (\overline{Q}_{1},\ldots,\overline{Q}_{n+1},\overline{P}_{1},\ldots,\overline{P}_{n+1})
			\end{eqnarray*}
			By $\overline{Q}_{j}$ and $\overline{P}_{j}$ we denote the complex conjugate of $Q_{j}$ and $P_{j}$, respectively.
			In \cite{Novak}, following the general introduction of the real forms of complex Hamiltonian systems in \cite{Gerdjikov}, the real forms of the complex Neumann system are introduced.
			
			Let ${\cal S} \subseteq \{1,2,\ldots,n+1\}$  be a nonempty subset and $\epsilon_{{\cal S},j}$ the discrete function defined in (\ref{eq:epsilon}).
			Denote by $J_{\cal S}$ the diagonal matrix with $1$ on positions ${\cal S}$ and $-1$ elsewhere.
			Other real forms are obtained as real forms of the antiholomorphic involutive automorphisms
			\begin{eqnarray}
			\tau_{\cal S}: T^{*}{\mathbb C}^{n+1} & \longrightarrow & T^{*}{\mathbb C}^{n+1} \nonumber \\
			(Q,P)             & \longmapsto     & (J_{\cal S}\overline{Q}, J_{\cal S}\overline{P}) \,. \label{eq:tau}
			\end{eqnarray}
			According to $\tau_{\cal S}$, we denote the corresponding real system as the triple $(T^{*}{\cal H}_{\cal S}^{n}, \omega_{\cal S},H_{\cal S})$, where $T^{*}{\cal H}_{\cal S}^{n}=\{(q,p)\in T^{*}{\mathbb R}^{n+1}\, \big| \; \sum_{j=1}^{n+1}\epsilon_{{\cal S},j}q_{j}^{2}=1, \sum_{j=1}^{n+1}\epsilon_{{\cal S},j}q_{j}p_{j}=0 \}$, i.e. the cotangent bundle of the hyperboloid ${\cal H}_{\cal S}^{n}=\{q \in {\mathbb R}^{n+1}\, \big| \; \sum_{j=1}^{n+1}\epsilon_{{\cal S},j}q_{j}^{2}=1 \}$, $\omega_{\cal S}=\sum_{j=1}^{n+1}\epsilon_{{\cal S},j} dq_{j} \wedge dp_{j}$, and
			$$ H_{\cal S}(q,p) = \frac{1}{2}\left( \sum_{j=1}^{n+1}\epsilon_{{\cal S},j}q_{j}^{2}\sum_{j=1}^{n+1}\epsilon_{{\cal S},j}p_{j}^{2}-(\sum_{j=1}^{n+1}\epsilon_{{\cal S},j}q_{j}p_{j})^{2}\right)+\frac{1}{2}\left(2-\sum_{j=1}^{n+1}\epsilon_{{\cal S},j}q_{j}^{2}\right)\sum_{j=1}^{n+1}\epsilon_{{\cal S},j}a_{j}q_{j}^{2}.$$
			The topology of the Arnold-Liouville level sets is known only for two families of the real forms, i.e. for all subsets ${\cal S} = \{1,2,\ldots,k\}$ for $k \in \{1,2,\ldots,n+1\}$ and all subsets ${\cal S} = \{k\}$ for $k \in \{1,2,\ldots,n+1\}$ (see \cite{Novak}).
			
			The complex Neumann system is a special example of the Mumford system (see \cite{Mumford, Novak, Vanhaecke}). It is characterized by the Lax pair $(L^{\mathbb C}(\lambda),M^{\mathbb C}(\lambda))$ of $2 \times 2$ matrices where
			$$ L^{\mathbb C} = \left(\begin{matrix}
			V^{\mathbb C}(\lambda)  &  W^{\mathbb C}(\lambda) \\
			U^{\mathbb C}(\lambda)  &  -V^{\mathbb C}(\lambda)
			\end{matrix}\right)$$
			for suitable polynomials $U^{\mathbb C}(\lambda)$, $V^{\mathbb C}(\lambda)$ and $W^{\mathbb C}(\lambda)$.
			It is known that the topology of an Arnold-Liouville level set of a real form directly depends on the positions of roots of the suitable real form of $U^{\mathbb C}(\lambda)$ according to the roots of suitable real form of $f^{\mathbb C}(\lambda)=U^{\mathbb C}(\lambda)W^{\mathbb C}(\lambda)+(V^{\mathbb C}(\lambda))^{2}$, and indirectly depends on the existence and positions of real roots of suitable real form of $U^{\mathbb C}(\lambda)$ according to the given constants $a_{1},a_{2},\ldots,a_{n+1}$.
			{\color{black} The topology of Arnold-Liouville level set of the Neumann system is   described in detail elsewhere, see c.f. \cite{Mumford} and \cite{Dullin2}.}
			
			{\color{black}
				Summarizing results in \cite{Vanhaecke} (pages 86-89), the polynomial
				$$ f_{\cal S}(\lambda) = U_{\cal S}(\lambda)W_{\cal S}(\lambda)+(V_{\cal S}(\lambda))^2=\prod_{i=1}^{n+1}(\lambda-a_i) \cdot \prod_{j=1}^{n}(\lambda-b_j) $$
				and roots $\lambda_1, \lambda_2,\ldots,\lambda_n$ of $U_{\cal S}(\lambda)$ determine the topology of the Arnold-Liouville level set. Therefore this can be either a torus or a cylinder or a disc with $g({\cal C})-1$ holes. The number of holes depends on the genus $g(\cal C)$ of hyperelliptic curve ${\cal C}: \mu^2+f(\lambda)=0$.}
			
			
			\bigskip
			
			\section{Positive solution of a system of linear equations}
			
			In \cite{Dines}, an algorithm for determining whether a given system of linear equations with real coefficients admits a positive solution is given. In this section, the idea of the algorithm is summarized with emphasis on a solution of a system of nonhomogeneous equations.
			
			\subsection{A single non-homogeneous equation}
			
			Assume that a non-homogeneous linear equation is given
			\begin{equation*}
			b_{1}x_{1}+b_{2}x_{2}+\ldots+b_{n}x_{n}=b \,.
			\end{equation*}
			For $b_{n+1}=-b$ and introducing the condition $x_{n+1}=1$, the equivalent homogeneous equation has a form
			\begin{equation}
			b_{1}x_{1}+b_{2}x_{2}+\ldots+b_{n}x_{n}+b_{n+1}x_{n+1} = 0 \, .
			\label{eq:homlineq3}
			\end{equation}
			
			Now, temporarily forget both conditions above.  A solution of homogeneous equation (\ref{eq:homlineq3}) can be found by the procedure described in \cite{Dines} that is briefly recalled below. Assume that some of coefficients in (\ref{eq:homlineq3}) are positive and some of coefficients are negative (which is a necessary condition for the existence of a positive solution for (\ref{eq:homlineq3})). The coefficients can be divided into two sets
			\begin{eqnarray*}
				{\rm positive \, coeficients}: & &  \{ b_{i} \} \quad {\rm for} \; i \in I = \{i_{1},i_{2},\ldots,i_{P}\} = \{i_{p} \,|\;p=1,2,\ldots,P\} \\
				{\rm negative \, coeficients}: & &  \{ b_{j} \} \quad {\rm for} \; j \in J = \{j_{1},j_{2},\ldots,j_{Q}\} = \{j_{q} \,|\;q=1,2,\ldots,Q\} \,.
			\end{eqnarray*}
			Note that $P+Q=n+1$. As soon as the equation (\ref{eq:homlineq3}) is written in the form
			\begin{equation*}
			\sum_{i \in I}b_{i}x_{i} = -\sum_{j \in J}b_{j}x_{j}
			\end{equation*}
			one among the positive solutions is obvious
			\begin{eqnarray*}
				x_{i} & = & -\sum_{j \in J} b_{j}   \qquad{\rm for} \; i \in I \\
				x_{j} & = & \;\;\; \sum_{i \in I} b_{i}   \qquad{\rm for} \; j \in J \,.
			\end{eqnarray*}
			

			\subsection{A system of non-homogeneous linear equations}
			
			Consider first the system of $m$ homogeneous equations given by
			\begin{equation}
			\sum_{s=1}^{n+1} b_{r,s}x_{s} = 0    \qquad ; \quad r=1,2,\ldots,m \,.
			\label{system:m}
			\end{equation}
			As in the previous subsection, the coefficients of the first equation can be separated into two subsets
			\begin{eqnarray*}
				{\rm positive \, coeficients}: & &  \{ b_{1,i} \} \quad {\rm for} \; i \in I=\{i_{1},i_{2},\ldots,i_{P}\} =  \{i_{p} \,|\;p=1,2,\ldots,P\}\\
				{\rm negative \, coeficients}: & &  \{ b_{1,j} \} \quad {\rm for} \; j \in J=\{j_{1},j_{2},\ldots,j_{Q}\} =  \{j_{q} \,|\;q=1,2,\ldots,Q\}\,,
			\end{eqnarray*}
			where $P+Q=n+1$, and the first equation can be written in the form
			\begin{equation}
			\!\!\!\!\!\!\!\!\!\!\!\!\!\!\!\!\!\!\!\!\!\!\!\!\!\sum_{i \in I}b_{1,i}x_{i} = -\sum_{j \in J}b_{1,j}x_{j} \,.
			\label{eq:1}
			\end{equation}
			Then, the remaining equations have to be written in the form
			\begin{equation}
			\qquad\qquad\qquad   \sum_{j \in J}b_{r,j}x_{j} = -\sum_{i \in I}b_{r,i}x_{i}   \qquad ; \quad r=2,3,\ldots,m \,.
			\label{eq:m-1}
			\end{equation}
			We multiply both sides of each equation of (\ref{eq:m-1}) by the corresponding sides of (\ref{eq:1}) and obtain $m-1$ equations
			\begin{equation}
			\sum_{i\in I,j\in J} b_{r,i,j}x_{i,j} = 0 \qquad {\rm for} \; r=2,3,\ldots,m ,
			\label{system:m-1}
			\end{equation}
			where
			\begin{equation}
			b_{r,i,j}=b_{1,i}b_{r,j}-b_{1,j}b_{r,i} \qquad {\rm and} \qquad x_{i,j}=x_{i}x_{j}\,.
			\label{eq:tripledoubleindex}
			\end{equation}
			Recall the next theorem that is proven in \cite{Dines}.
			\begin{theorem}[Dines \cite{Dines}] To every positive solution of the system (\ref{system:m}) there corresponds a positive solution of the system (\ref{system:m-1}), and conversely.
			\end{theorem}
			
			If a positive solution $(\overline{x}_{1},\ldots,\overline{x}_{n+1})$ of the system $(\ref{system:m})$ is given, the solution of the system (\ref{system:m-1}) is obvious. Conversely, if the system (\ref{system:m-1}) admits a positive solution $(\overline{x}_{i,j})$, then
			\begin{equation*}
			x_{i} = -\sum_{j\in J}b_{1,j}\overline{x}_{i,j} , \qquad {\rm and} \qquad x_{j}=\sum_{i \in I}b_{1,i}\overline{x}_{i,j} \,.
			\end{equation*}
			
			Next corollary is related to the existence of a positive solution of a system of non-homogeneous linear equations.
			\begin{corollary} To every positive solution of the system (\ref{system:m}) where $x_{n+1}=1$ there corresponds a positive solution of the system (\ref{system:m-1}), and conversely.
			\end{corollary}
			
			\begin{proof}
				If $(\overline{x}_{1},\ldots,\overline{x}_{n},1)$ is a positive solution of the system (\ref{system:m}), then $(\overline{x}_{1},\ldots,\overline{x}_{n},1)$ also satisfies (\ref{eq:1}) and (\ref{eq:m-1}). Therefore, $x_{i,j}=\overline{x}_{i}\overline{x}_{j}$ for $i,j \in \{1,\ldots,n\}$ and $x_{i,n+1}=x_{n+1,i}=\overline{x}_{i}$ satisfy (\ref{system:m-1}).
				
				Conversely, let  $(\overline{x}_{i,j})$ be a positive solution of the system (\ref{system:m-1}). Without the condition on the value of $x_{n+1}$, a positive solution of (\ref{system:m}) is known from the above theorem, i.e.
				\begin{equation*}
				x_{i} = -\sum_{j \in J}b_{1,j}\overline{x}_{i,j} , \qquad {\rm and} \qquad x_{j}=\sum_{i \in I}b_{1,i}\overline{x}_{i,j} \,.
				\end{equation*}
				Without loss of generality, set $n+1 \in J$. Since $P \neq 0$, $Q \neq 0$ and $\overline{x}_{i,j} >0$, we have
				\begin{equation}
				\sum_{i \in I} b_{1,i}\overline{x}_{i,n+1} > 0.
				\label{sum:norm}
				\end{equation}
				Hence, we can "normalize" the given solution of the system (\ref{system:m}) by dividing the solution above by the sum (\ref{sum:norm})  as follows
				\begin{equation*}
				x_{n+1}  =  \frac{\sum_{i \in I}b_{1,i}\overline{x}_{i,n+1}}{\sum_{i \in I}b_{1,i}\overline{x}_{i,n+1}} \; = \; 1, \qquad
				x_{j}    =  \frac{\sum_{i \in I}b_{1,i}\overline{x}_{i,j}}{\sum_{i \in I}b_{1,i}\overline{x}_{i,n+1}} \qquad
				{\rm and} \qquad
				x_{i}    =  -\frac{\sum_{j \in J}b_{1,j}\overline{x}_{i,j}}{\sum_{i \in I}b_{1,i}\overline{x}_{i,n+1}} \, .
				\end{equation*}
				
			\end{proof}
			
			
			\section{Algorithm}
			
			First, we choose an arbitrary nonempty subset ${\cal S} \subseteq \{1,2,\ldots,n+1\}$ and write the system (\ref{eq:system}) in the homogeneous form
			\begin{equation}
			\sum_{s=1}^{n+2} b_{r,s}q_{s}^{2} = 0 \qquad ; \quad r=0,1,\ldots,n \, ,
			\label{eq:systemformal}
			\end{equation}
			where $b_{0,s}=\epsilon_{{\cal S},s}$ for $s=1,2,\ldots,n+1$ (see (\ref{eq:epsilon})), $b_{0,n+2}=-1$, $b_{r,s}=\epsilon_{{\cal S},s}\frac{A(\lambda_{r})}{\lambda_{r}-a_{s}}$ for $s=1,2,\ldots,n+1$ and $r=1,2,\ldots,n$, and $b_{r,n+2}=0$ for $r=1,2,\ldots,n$.

			\subsection{The basic concept}
			
			\noindent
			{\bf Set}  the positions of the parameters $\lambda_{1}<\lambda_{2}<...<\lambda_{n}$ with respect to the values \ $a_{1}<a_{2}<...<a_{n+1}$. \\
			
			\noindent
			{\bf Step 1.} Verify the existence of the solution (every equation contains both positive and negative coefficients). If there exists one equation with all positive or all negative coefficients, go to {\bf Step 3}. Else, if there is only one equation, go to {\bf Step 4}, else, go to {\bf Step 2}. \\
			
			\noindent
			{\bf Step 2.} Set $k$ as the number of equations. In the first equation, separate the values; the positive values on the left side, the negative
			values on the right side. In all other equations, separate the values by the indices as: on the left side stay indices which go on the right side in the first equation, on the right side go indices which stay on the left side in the first equation. Form a new linear homogeneous system as: \\
			a) the left side of the $j$-th equation is obtained by multiplying the left side of the first equation (in the original system) and the left side of the $(j+1)$-th equation (in the original system), the right side of the $j$-th equation is obtained by multiplying the right side of the first equation (in the original system) and the right side of the $(j+1)$-th equation (in the original system)
			\\
			b) arranging new equations into the homogeneous linear system, new variables are the products of the previous ones.
			\\
			
			\noindent
			{\bf Step 3.} The system does not have any positive solution. \\
			
			\noindent
			{\bf Step 4.} The system has a positive solution.
			
			\subsection{Example $n=2$ and ${\cal S}=\{1,3\}$}

			Let $a_{1} < a_{2} < a_{3}$ and $(q_{1},q_{2},q_{3})$  be coordinates in ${\mathbb R}^{3}$. Take ${\cal S}=\{1,3\}$ (see Example \ref{ex:11} and Example 2.1. in \cite{Novak}). We would like to characterize the roots of the polynomial
			\begin{equation*}
			U_{\{1,3\}}(\lambda) = (\lambda-a_1)(\lambda-a_2)(\lambda-a_3)\left(\frac{q_{1}^{2}}{\lambda-a_{1}} - \frac{q_{2}^{2}}{\lambda-a_{2}} + \frac{q_{3}^{2}}{\lambda-a_{3}}\right)
			\end{equation*}
			according to the constraint $q_{1}^{2}-q_{2}^{2}+q_{3}^{2}=1$. There are 10 possible positions of real roots $\lambda_{1}$, $\lambda_{2}$ (according to the ordered constants $a_{1}$, $a_{2}$, $a_{3}$) of the polynomial $U_{\{1,3\}}$:
			\begin{enumerate}[(i)]
				\item $\lambda_1<\lambda_2<a_1<a_2<a_3$, \label{ex:i}
				\item $\lambda_1<a_1<\lambda_2<a_2<a_3$,
				\item $\lambda_1<a_1<a_2<\lambda_2<a_3$,
				\item $\lambda_1<a_1<a_2<a_3<\lambda_2$,
				\item $a_1<\lambda_1<\lambda_2<a_2<a_3$,
				\item $a_1<\lambda_1<a_2<\lambda_2<a_3$,
				\item $a_1<\lambda_1<a_2<a_3<\lambda_2$,
				\item $a_1<a_2<\lambda_1<\lambda_2<a_3$,
				\item $a_1<a_2<\lambda_1<a_3<\lambda_2$,
				\item $a_1<a_2<a_3<\lambda_1<\lambda_2$.
			\end{enumerate}
			They dictate $10$ verifications of the solvability of the system
			\begin{eqnarray}
			q_{1}^{2} - q_{2}^{2} + q_{3}^{2} & = & 1  \nonumber \\
			(\lambda_1-a_2)(\lambda_1-a_3)q_{1}^{2} - (\lambda_1-a_1)(\lambda_1-a_3)q_{2}^{2} + (\lambda_1-a_1)(\lambda_1-a_2)q_{3}^{2} & = & 0  \label{system:1,3}\\
			(\lambda_2-a_2)(\lambda_2-a_3)q_{1}^{2} - (\lambda_2-a_1)(\lambda_2-a_3)q_{2}^{2} + (\lambda_2-a_1)(\lambda_2-a_2)q_{3}^{2} & = & 0  \,. \nonumber
			\end{eqnarray}
			The system (\ref{system:1,3}) is equivalent to the homogeneous system of four variables
			\begin{eqnarray}
			q_{1}^{2} - q_{2}^{2} + q_{3}^{2} - q_{4}^{2} & = & 0  \nonumber \\
			(\lambda_{1}-a_{2})(\lambda_{1}-a_{3})q_{1}^{2} - (\lambda_{1}-a_{1})(\lambda_{1}-a_{3})q_{2}^{2} + (\lambda_{1}-a_{1})(\lambda_{1}-a_{2})q_{3}^{2} + 0 \cdot q_{4}^{2} & = & 0 \label{system:1,3homo}\\
			(\lambda_{2}-a_{2})(\lambda_{2}-a_{3})q_{1}^{2} - (\lambda_{2}-a_{1})(\lambda_{2}-a_{3})q_{2}^{2} + (\lambda_{2}-a_{1})(\lambda_{2}-a_{2})q_{3}^{2} + 0 \cdot q_{4}^{2} & = & 0 \nonumber
			\end{eqnarray}
			with the requirement $q_{4}=1$. In the notation of the homogeneous system $\sum_{s=1}^{4} b_{r,s}q_{s}^{2} = 0$ ; \ $r=0,1,2$ (see (\ref{eq:systemformal})) we have
			$$\begin{array}{llll}
			b_{0,1} = 1 & \ \ \ b_{0,2} = -1 & \ \ \ b_{0,3} = 1 & \ \ \ b_{0,4} = -1 \\
			b_{1,1}=(\lambda_{1}-a_{2})(\lambda_{1}-a_{3})  & \ \ \ b_{1,2} = -(\lambda_{1}-a_{1})(\lambda_{1}-a_{3})  & \ \ \ b_{1,3} = (\lambda_{1}-a_{1})(\lambda_{1}-a_{2})  & \ \ \ b_{1,4} = 0  \\
			b_{2,1}=(\lambda_{2}-a_{2})(\lambda_{2}-a_{3})  & \ \ \ b_{2,2} = -(\lambda_{2}-a_{1})(\lambda_{2}-a_{3})  & \ \ \ b_{2,3} = (\lambda_{2}-a_{1})(\lambda_{2}-a_{2})  & \ \ \ b_{2,4} = 0 \,.
			\end{array}$$
			To illustrate the algorithm procedure take one of ten possibilities of the roots positions, for example case (\ref{ex:i}), i.e. $\lambda_{1} < \lambda_{2} < a_{1} < a_{2} < a_{3}$. Hence $b_{1,1}>0$, $b_{1,2}<0$, $b_{1,3}>0$ and $b_{2,1}>0$, $b_{2,2}<0$, $b_{2,3}>0$ and this (together with the obvious positive and negative coefficients in the first line) allows a positive solution of the system. We reorganize the first equation as proposed in (\ref{eq:1})
			\begin{eqnarray*}
				b_{0,1}q_{1}^{2} + b_{0,3}q_{3}^{2} & = & -b_{0,2}q_{2}^{2} - b_{0,4}q_{4}^{2},
			\end{eqnarray*}
			and with respect to the first equation the remaining two (see (\ref{eq:m-1})) are
			\begin{eqnarray*}
				b_{1,2}q_{2}^{2} + b_{1,4}q_{4}^{2} & = & -b_{1,1}q_{1}^{2} - b_{1,3}q_{3}^{2} \\
				b_{2,2}q_{2}^{2} + b_{2,4}q_{4}^{2} & = & -b_{2,1}q_{1}^{2} - b_{2,3}q_{3}^{2} \,.
			\end{eqnarray*}
			According to the notation in Section 2, indices of positive and negative coefficients are
			$I=\{i_{1},i_{2}\}=\{ 1,3 \}$, $J=\{j_{1},j_{2}\}=\{2,4\}$ and $P=Q=2$. By multiplying we get the new system of two linear equations (see (\ref{system:m-1})). Orderly, we can write
			\begin{eqnarray}
			(b_{0,1}b_{1,2}-b_{0,2}b_{1,1})q_{1}^{2}q_{2}^{2} + (b_{0,1}b_{1,4}-b_{0,4}b_{1,1})q_{1}^{2}q_{4}^{2} + (b_{0,3}b_{1,2}-b_{0,2}b_{1,3})q_{2}^{2}q_{3}^{2}+ &  & \nonumber \\
			+ (b_{0,3}b_{1,4}-b_{0,4}b_{1,3})q_{3}^{2}q_{4}^{2} & = & 0 \nonumber\\
			(b_{0,1}b_{2,2}-b_{0,2}b_{2,1})q_{1}^{2}q_{2}^{2} + (b_{0,1}b_{2,4}-b_{0,4}b_{2,1})q_{1}^{2}q_{4}^{2} + (b_{0,3}b_{2,2}-b_{0,2}b_{2,3})q_{2}^{2}q_{3}^{2}+ &  &  \label{system:(1)1}\\
			+ (b_{0,3}b_{2,4}-b_{0,4}b_{2,3})q_{3}^{2}q_{4}^{2} & = & 0 \nonumber
			\end{eqnarray}
			We check the signs of coefficients of the first equation (taking into account the order of $\lambda_{1}$, $\lambda_{2}$, $a_{1}$, $a_{2}$, $a_{3}$):
			\begin{eqnarray*}
				b_{0,1}b_{1,2}-b_{0,2}b_{1,1} & = & 1 \cdot(-(\lambda_{1}-a_{1})(\lambda_{1}-a_{3}))+1\cdot(\lambda_{1}-a_{2})(\lambda_{1}-a_{3}) = \\
				& = & (\lambda_{1}-a_{3})(a_{1}-a_{2})
			\end{eqnarray*}
			therefore $\sgn(b_{0,1}b_{1,2}-b_{0,2}b_{1,1}) = 1$,
			\begin{eqnarray*}
				b_{0,1}b_{1,4}-b_{0,4}b_{1,1} & = & 0 - (-1)(\lambda_{1}-a_{2})(\lambda-a_{3}) = (\lambda_{1}-a_{2})(\lambda_{1}-a_{3})
			\end{eqnarray*}
			and therefore $\sgn(b_{0,1}b_{1,4}-b_{0,4}b_{1,1}) = 1$,
			\begin{eqnarray*}
				b_{0,3}b_{1,2}-b_{0,2}b_{1,3} & = & -(\lambda_{1}-a_{1})(\lambda_{1}-a_{3})+(\lambda_{1}-a_{1})(\lambda_{1}-a_{2}) = \\
				& = & (\lambda_{1}-a_{1})(-\lambda_{1}+a_{3}+\lambda_{1}-a_{2}) \;\; = \;\; (\lambda_{1}-a_{1})(a_{3}-a_{2})
			\end{eqnarray*}
			thus $\sgn(b_{0,3}b_{1,2}-b_{0,2}b_{1,3}) = -1$, and
			\begin{eqnarray*}
				b_{0,3}b_{1,4}-b_{0,4}b_{1,3} & = & 0 + (\lambda_{1}-a_{1})(\lambda_{1}-a_{2}) \;\; = \;\; (\lambda_{1}-a_{1})(\lambda_{1}-a_{2})
			\end{eqnarray*}
			and it follows $\sgn(b_{0,3}b_{1,4}-b_{0,4}b_{1,3})=1$.
			
			Similarly, we can calculate the signs of the coefficients of the second equation. We realize that the solution of the system may occur.
			We use the triple indexation (see \cite{Dines} and (\ref{eq:tripledoubleindex}))
			\begin{equation*}
			b_{r,i_{p},j_{q}} = b_{0,i_{p}}b_{r,j_{q}}-b_{0,j_{q}}b_{r,i_{p}} \quad {\rm for} \; r=1,2  \,.
			\end{equation*}
			Hence, the triple indexation can be switch to the double indexation as follows
			\begin{equation*}
			b_{r,i_{p},j_{q}} = b^{(1)}_{r-1,(q-1)P+p} \qquad {\rm and} \qquad
			q_{i}^{2}q_{j}^{2} = q_{i,j}^{2} = q_{i_{p},j_{q}}^{2} \,.
			\end{equation*}
			As above (for the coefficients) the double indexation in variables can become single by
			\begin{equation*}
			q_{i_{p},j_{p}}^{2} = q_{(q-1)P+p}^{2} \,.
			\end{equation*}
			Then, the system (\ref{system:(1)1}) can be simply written as
			\begin{eqnarray*}
				b_{0,1}^{(1)}q_{1}^{2} + b_{0,2}^{(1)}q_{2}^{2} + b_{0,3}^{(1)}q_{3}^{2} + b_{0,4}^{(1)}q_{4}^{2} & = & 0 \nonumber\\
				b_{1,1}^{(1)}q_{1}^{2} + b_{1,2}^{(1)}q_{2}^{2} + b_{1,3}^{(1)}q_{3}^{2} + b_{1,4}^{(1)}q_{4}^{2} & = & 0 \,.
			\end{eqnarray*}
			Thus, the procedure has to be repeated for the smaller system of equations.
			With respect to the signs of the coefficient the first equation has to be written as
			\begin{equation*}
			b_{0,1}^{(1)}q_{1}^{2} + b_{0,2}^{(1)}q_{2}^{2} + b_{0,4}^{(1)}q_{4}^{2} = -b_{0,3}^{(1)}q_{3}^{2}
			\end{equation*}
			and accordingly the second equation as
			\begin{equation*}
			b_{1,3}^{(1)}q_{3}^{2} = -b_{1,1}^{(1)}q_{1}^{2} - b_{1,2}^{(1)}q_{2}^{2} - b_{1,4}^{(1)}q_{4}^{2} \,.
			\end{equation*}
			By multiplying (left side with left side and right side with right side) and reordering, we obtain the last linear equation in the procedure
			\begin{equation*}
			\left(b_{0,1}^{(1)}b_{1,3}^{(1)}-b_{0,3}^{(1)}b_{1,1}^{(1)}\right)q_{1}^{2}q_{3}^{2} + \left(b_{0,2}^{(1)}b_{1,3}^{(1)}-b_{0,3}^{(1)}b_{1,2}^{(1)}\right)q_{2}^{2}q_{3}^{2} + \left(b_{0,4}^{(1)}b_{1,3}^{(1)}-b_{0,3}^{(1)}b_{1,4}^{(1)}\right)q_{3}^{2}q_{4}^{2} \; = \; 0 \,.
			\end{equation*}
			By straightforward but tedious calculation the first and the third coefficient can be factorized\footnote{in Mathematica one can use the functions Assume and Factor:\\ Assume[$\lambda_1<\lambda_2<a_1<a_2<a_3$, Factor[$(-(\lambda_1-a_1)(\lambda_1-a_3)+(\lambda_1-a_2)(\lambda_1-a_3))(-(\lambda_2-a_1)(\lambda_2-a_3)+(\lambda_2-a_1)(\lambda_2-a2))-(-(\lambda_1-a_1)(\lambda_1-a_3)+(\lambda_1-a_1)(\lambda_1-a_2))(-(\lambda_2-a_1)(\lambda_2-a_3)+(\lambda_2-a_2)(\lambda_2-a_3))$]]}\textsuperscript{,}\footnote{In MATLAB: assume($\lambda_1<\lambda_2<a_1<a_2<a_3$); factor($(-(\lambda_1-a_1)(\lambda_1-a_3)+(\lambda_1-a_2)(\lambda_1-a_3))(-(\lambda_2-a_1)(\lambda_2-a_3)+(\lambda_2-a_1)(\lambda_2-a2))-(-(\lambda_1-a_1)(\lambda_1-a_3)+(\lambda_1-a_1)(\lambda_1-a_2))(-(\lambda_2-a_1)(\lambda_2-a_3)+(\lambda_2-a_2)(\lambda_2-a_3))$)} as
			\begin{equation*}
			b_{0,1}^{(´1)}b_{1,3}^{(1)}-b_{0,3}^{(1)}b_{11}^{(1)} = (a_{2}-a_{1})(a_{3}-a_{1})(a_{3}-a_{2})(\lambda_{2}-\lambda_{1})
			\end{equation*}
			and
			\begin{equation*}
			b_{0,4}^{(1)}b_{1,3}^{(1)} - b_{0,3}^{(1)}b_{1,4}^{(1)} = -(a_{3}-a_{2})(a_{1}-\lambda_{1})(a_{1}-\lambda_{2})(\lambda_{2}-\lambda_{1})\,.
			\end{equation*}
			Thus, the signs of both coefficients are
			\begin{equation*}
			\sgn(b_{0,1}^{(´1)}b_{1,3}^{(1)}-b_{0,3}^{(1)}b_{1,1}^{(1)}) = -1
			\end{equation*}
			and
			\begin{equation*}
			\sgn(b_{0,4}^{(1)}b_{1,3}^{(1)} - b_{0,3}^{(1)}b_{1,4}^{(1)})=1,
			\end{equation*}
			and hence a positive solution of system (\ref{system:1,3}) exists.
			
			Since all these tedious calculations are time-consuming, especially for large systems, the need for use of computer computation is obvious. For writing a code in MATLAB, matrix forms of the systems are needed.
			
			\subsection{Matrix form of algorithm}
			
			For reasons of transparency, it would make sense to write the system of equations (\ref{system:m}) (or (\ref{eq:systemformal}) for an appropriate number of indices) in the matrix form. To a system of $m$ homogeneous linear equations (\ref{system:m}) corresponds the matrix of the system
			\begin{equation*}
			B = [ b_{r,s}]_{\substack{r=1,2,...,m \\s=1,2,...,n+1}} \,.
			\end{equation*}
			In this discussion, coefficients $b_{r,s}$ can be arbitrary real numbers.
			According to the separation of equations (see (\ref{eq:1}) and (\ref{eq:m-1})) four submatrices of dimensions $1 \times P$, $1 \times Q$, $(m-1) \times P$ and $(m-1) \times Q$ are defined respectively as
			\begin{eqnarray}
			C_{11}  & := & [b_{1,i_{k}}]_{k=1,2,...,P}
			\label{eq:C11} \\
			C_{12}  & := & [b_{1,j_{k}}]_{k=1,2,...,Q}
			\label{eq:C12} \\
			C_{21}  & := & [b_{r,i_{k}}]_{\substack{r=2,3,...,m \\ k=1,2,...,P}}
			\label{eq:C21} \\
			C_{22}  & := & [b_{r,j_{k}}]_{\substack{r=2,3,...,m \\ k=1,2,...,Q}}
			\label{eq:C22}
			\end{eqnarray}
			Define two operations on the set of $m \times t$ matrix $M^{m,t}$. The first one returns the $i$-th row of a matrix
			\begin{eqnarray*}
				R_{i} : \qquad \,\,  M^{m,t} \quad & \longrightarrow & M^{1,t} \\
				R_{i} : [a_{r,s}]_{\substack{r=1,2,...,m \\ s=1,2,...,t}} & \longmapsto & [a_{i,s}]_{s=1,2,...,t} \, .
			\end{eqnarray*}
			The second one converts the matrix into a column vector (vectorization) as
			\begin{eqnarray*}
				\mtov : \qquad \qquad M^{m,t} & \longrightarrow & M^{mt,1} \\
				\mtov :   \;\;\;\; [a_{r,s}]_{\substack{r=1,2,...,m \\ s=1,2,...,t}} & \longmapsto & [a_{1,1},a_{2,1},\ldots,a_{m,1},a_{1,2},a_{2,2},\ldots,a_{m,2},\ldots,a_{1,t},a_{2,t},\ldots,a_{m,t}]^{T} \,.
			\end{eqnarray*}
			
			\begin{proposition}
				\label{prop:reducedmatrix}
				The matrix of the reduced system (\ref{system:m-1}) of the main system (\ref{system:m}) is the $(m-1) \times PQ$ matrix $D$, which rows are
				\begin{equation*}
				R_{i}(D) = \mtov \left(C_{11}^{T} \cdot R_{i}(C_{22})-(R_{i}(C_{21}))^{T} \cdot C_{12} \right)
				\end{equation*}
				for $i=1,2,\ldots,m-1$.
			\end{proposition}
			\begin{proof}
				By direct calculation we obtain \\
				
				$\displaystyle{C_{11}^{T} \cdot R_{i}(C_{22})-(R_{i}(C_{21}))^{T} \cdot C_{12}} = $
				\begin{eqnarray*}
					& = &
					\left[\begin{matrix}
						b_{1,i_{1}} \\ b_{1,i_{2}} \\ \vdots \\ b_{1,i_{P}}
					\end{matrix}\right]
					\cdot
					\left[\begin{matrix}
						b_{i,j_{1}} & b_{i,j_{2}} & \ldots & b_{i,j_{Q}}
					\end{matrix}\right]
					-
					\left[\begin{matrix}
						b_{i,i_{1}} \\ b_{i,i_{2}} \\ \vdots \\ b_{i,i_{P}}
					\end{matrix}\right]
					\cdot
					\left[\begin{matrix}
						b_{1,j_{1}} & b_{1,j_{2}} & \ldots & b_{1,j_{Q}}
					\end{matrix}\right]
					\; = \\
					& = &
					\left[\begin{matrix}
						b_{1,i_{1}}b_{i,j_{1}}  &  b_{1,i_{1}}b_{i,j_{2}}  &  \ldots  &  b_{1,i_{1}}b_{i,j_{Q}} \\
						b_{1,i_{2}}b_{i,j_{1}}  &  b_{1,i_{2}}b_{i,j_{2}}  &  \ldots  &  b_{1,i_{2}}b_{i,j_{Q}} \\
						\vdots               &    \vdots                &  \ddots  &  \vdots                 \\
						b_{1,i_{P}}b_{i,j_{1}}  &  b_{1,i_{P}}b_{i,j_{2}}  &  \ldots  &  b_{1,i_{P}}b_{i,j_{Q}} \\
					\end{matrix}\right]
					-
					\left[\begin{matrix}
						b_{i,i_{1}}b_{1,j_{1}}  &  b_{i,i_{1}}b_{1,j_{2}}  &  \ldots  &  b_{i,i_{1}}b_{1,j_{Q}} \\
						b_{i,i_{2}}b_{1,j_{1}}  &  b_{i,i_{2}}b_{1,j_{2}}  &  \ldots  &  b_{i,i_{2}}b_{1,j_{Q}} \\
						\vdots               &    \vdots                &  \ddots  &  \vdots                 \\
						b_{i,i_{P}}b_{1,j_{1}}  &  b_{i,i_{P}}b_{1,j_{2}}  &  \ldots  &  b_{i,i_{P}}b_{1,j_{Q}} \\
					\end{matrix}\right].
				\end{eqnarray*}
				Therefore, all elements in the $(m-1) \times PQ$ matrix $\mtov \left(C_{11}^{T} \cdot R_{i}(C_{22})-(R_{i}(C_{21}))^{T} \cdot C_{12} \right)$ are in the form
				\begin{equation*}
				b_{1,i_{p}}b_{i,j_{q}} - b_{i,i_{p}}b_{1,j_{q}}
				\end{equation*}
				for $p \in \{1,2,\ldots,P\}$ and $q \in \{1,2,\ldots,Q\}$. These elements are exactly the coefficients of the reduced system (\ref{system:m-1}).
			\end{proof}

			\subsection{Algorithm for $n=2$}
			
			\label{subsection:alg n2}
			
			At the beginning of this section, the course of the algorithm for checking the existence of a positive solution of the system (\ref{eq:systemformal}) is given. Since the signs of coefficients of some reduced systems can not be easily determined in general (see examples in \ref{appendix:analysis}), we restricted our investigation to the low dimension $n=2$ and all polynomials $U_{\cal S}(\lambda)$ for nonempty subset ${\cal S} \subseteq \{1,2,3\}$. For coding the algorithm we use MATLAB. To recording the matrix of the main system for all nonempty subsets ${\cal S}$ and all possible positions of roots $\lambda_{1}$ and $\lambda_{2}$, according to the given values $a_{1}$, $a_{2}$ and $a_{3}$, $70$ scripts are prepared. Every script is denoted by $\rm{Snum_{1}Lnum_{2}}$, where
			
			$\bullet$ $\rm{num_{1}}$ denotes the subset ${\cal S}$ as: $\rm{num_{1}}=1$ determines $\{1\}$, $\rm{num_{1}}=12$ determines $\{1,2\}$, etc.
			
			$\bullet$ $\rm{num_{2}}$ denotes the positions of roots $\lambda_{1}$ and $\lambda_{2}$ respectively: $0$ means that a root lies in $(-\infty,a_{1})$,
			the
			
			\hspace{0.15cm} value $1$ means that a root lies in $(a_{1},a_{2})$, etc. (in MATLAB, roots $\lambda_{1}$ and $\lambda_{2}$ are denoted by $l_{1}$ and $l_{2}$).
			
			\noindent
			Every script determines two matrices, System and SignsSystem. The matrix System is the matrix of system (\ref{system:1,3homo}) and SignSystem is the matrix of signs of the elements in System, respectively. For the visualization, see \ref{appendix:scripts and fncs}, Script ${\rm S1L00}$ and Script ${\rm S13L12}$.
			
			The algorithm contains two essential functions. The first one is {\ttfamily{eops}} (existence of a positive solution). Its inputs are the matrix $B = \rm{SignsSystem}$ and the integer $n=$ ((the number of equations)$-1$). Its output is a vector $D$, which components are
			
			$\bullet$ the first component is $1$ and the second component is an empty vector if a positive solution exists. In
			
			\hspace{0.15cm} MATLAB, this is equal to value 1.
			
			$\bullet$ the first component is $0$ and the second component is the row of the matrix $B=\rm{SignsSystem}$, therefore
			
			\hspace{0.15cm} the system does not allow a positive solution.
			
			\noindent
			The second one is {\ttfamily{newsys}} function. Its inputs are submatrices of the matrix of the system, i.e. matrices (\ref{eq:C11}), (\ref{eq:C12}), (\ref{eq:C21}) and (\ref{eq:C22}), respectively. By the {\bf for} loop function {\ttfamily{newsys}} constructs the matrix of the reduced system as is presented in Proposition \ref{prop:reducedmatrix}. The output of {\ttfamily{newsys}} function is the matrix of the reduced system {\ttfamily{Sys}} and the matrix of signs {\ttfamily{B}} of elements in {\ttfamily{Sys}}, respectively.
			
			For all 70 scripts, the algorithm confirms or denies the existence of a positive solution of the inputted system. All results are presented in \ref{appendix:results of alg} (Table \ref{Tab:1}, Table \ref{Tab:2} and Table \ref{Tab:3}). There, each subtable contains three columns. In the first column, there are the names of the scripts. In the $Y/N$ column, 1 declares the existence of a positive solution, and 0 declares the non-existence of a positive solution. The third column $step$ tells us in which step of the algorithm the existence of a positive solution is denied. Additional $pf$ along with the number means that after running the algorithm it returns a row of $B=SignsSys$ which values of elements ($1$ and $-1$) can not be determined by a factorization. The elements remain expressions in variables $\lambda_{1}$, $\lambda_{2}$, $a_{1}$, $a_{2}$ and $a_{3}$. In \ref{appendix:analysis}, all these examples are further analyzed. After checking, by exhaustive calculations, it is confirmed that all examples can be partially factorized and the signs of expressions can be determined.

			\section{Conclusion}
			
			The designed algorithm confirms the known facts about the roots of the polynomials $U_{\cal S}(\lambda)$ for all nonempty sets ${\cal S} \subseteq\{1,2,3\}$ and ${\cal S} \neq \{1,3\}$. For ${\cal S}=\{1,3\}$, a new result is obtained. We realize that the polynomial $U_{\{1,3\}}(\lambda)$ has real roots for some $(q_{1},q_{2},q_{3}) \in {\cal H}_{\{1,3\}}^{2}$. Moreover, the new results of this work include possible positions of roots according to the values $a_{1}$, $a_{2}$ and $a_{3}$ (see Table \ref{Tab:2}, the subtable of subset $\{1,3\}$). \  More precisely, we show that there are points in the hyperboloid ${\cal H}_{\{1,3\}}^{2}$ at which the polynomial $U_{\{1,3\}}(\lambda)$ has real roots :
			\begin{enumerate}[(a)]
				\item for some $(q_{1},q_{2},q_{3}) \in {\cal H}_{\{1,3\}}^{2}$ polynomial $U_{\{1,3\}}(\lambda)$ has real roots $\lambda_1, \lambda_2 \in (-\infty, a_{1})$,
				\item for some $(q_{1},q_{2},q_{3}) \in {\cal H}_{\{1,3\}}^{2}$ polynomial $U_{\{1,3\}}(\lambda)$ has real roots $\lambda_1, \lambda_2 \in (a_{1}, a_{2})$,
				\item for some $(q_{1},q_{2},q_{3}) \in {\cal H}_{\{1,3\}}^{2}$ it has real roots $\lambda_1, \lambda_2 \in (a_{2}, a_{3})$ and
				\item for some $(q_{1},q_{2},q_{3}) \in {\cal H}_{\{1,3\}}^{2}$ it has real roots $\lambda_1, \lambda_2 \in (a_{3}, \infty)$.
			\end{enumerate}
			
			
			\noindent
			{\bf The topology of Arnold-Liouville level sets.}
			Following \cite{Mumford, Novak, Vanhaecke},
			let $$f_{\cal S}(\lambda) = U_{\cal S}(\lambda)W_{\cal S}(\lambda)+V_{\cal S}^{2}(\lambda) = -(\lambda-a_1)(\lambda-a_2)(\lambda-a_3)(\lambda-b_1)(\lambda-b_2)$$
			be a characteristic polynomial; i.e. the real form of $f^{\mathbb C}(\lambda)=U^{\mathbb C}(\lambda)W^{\mathbb C}(\lambda)+(V^{\mathbb C}(\lambda))$ according to a subset ${\cal S} \subseteq \{1,2,3\}$ and the automorphism $\tau_{\cal S}$ (see (\ref{eq:tau})). The topology of an Arnold-Liouville level set is determined by the position of roots $\lambda_{1}$, $\lambda_2$ of polynomial $U_{\cal S}(\lambda)$ with respect to the positions of the roots $a_{1},a_{2},a_{3},b_{1},b_{2}$ of polynomial $f_{\cal S}(\lambda)$. As we here consider only the example ${\cal S}=\{1,3\}$, let $f(\lambda)=f_{\{1,3\}}(\lambda)$.
			\\
			
			In all four cases below, the following properties are used
			\begin{equation}
			f(\lambda_1) > 0 \qquad \hbox{and} \qquad f(\lambda_2) > 0,
			\label{eq:prop1}
			\end{equation}
			\begin{equation}
			f(a_1)=f(a_2)=f(a_3)=0
			\label{eq:prop2}
			\end{equation}
			and
			\begin{equation}
			f(\lambda)<0 \quad \hbox{for} \;\; \lambda >\!\!>0 \qquad \hbox{and} \qquad f(\lambda)>0 \quad \hbox{for} \;\; \lambda <\!\!<0 \,.
			\label{eq:prop3}
			\end{equation}
			\noindent
			{\bf Case (a), $\lambda_{1},\lambda_{2} \in (-\infty,a_1)$:} The possibilities of the positions of the roots $b_1$ and $b_2$ of the polynomial $f(\lambda)$  are the following
			\begin{enumerate}[(i)]
				\item $b_1$ and $b_2$ are complex roots. Therefore
				$ f(\lambda) = -(\lambda-a_1)(\lambda-a_2)(\lambda-a_3)(\lambda^2+\beta^2) $
				for some $\beta \in {\mathbb R}$. The roots of the polynomials $U_{\{1,3\}}(\lambda)$ and $f(\lambda)$ are on positions as in Figure \ref{figA1}.
				\begin{figure}[h!]
					\centering
					\includegraphics[scale=0.6]{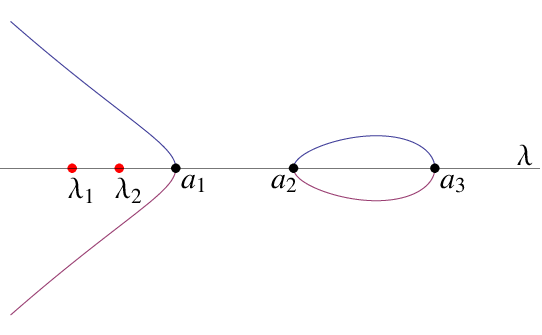}
					\caption{Real part of the spectral curve $\mu^{2}=f^{\mathbb C}(\lambda)$ in regard to the involution $\tau_{\{1,3\}}$ and case $\lambda_{1},\lambda_{2} \in (-\infty,a_1)$ and $b_1,b_2 \in {\mathbb C}$}
					\label{figA1}
				\end{figure}
				\item $b_1,b_2 \in (-\infty,\lambda_1)$. The roots $\lambda_1$ and $\lambda_2$ stay in the same interval $(b_2,a_1)$ (see Figure \ref{figA2}).
				\begin{figure}[h!]
					\centering
					\includegraphics[scale=0.6]{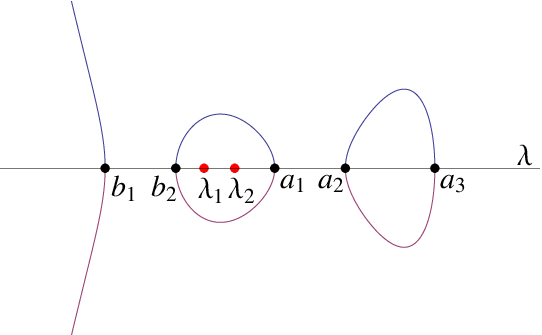}
					\caption{Real part of the spectral curve $\mu^{2}=f^{\mathbb C}(\lambda)$ in regard to the involution $\tau_{\{1,3\}}$ and case $\lambda_{1},\lambda_{2} \in (-\infty,a_1)$ and $b_1,b_2 \in (-\infty,\lambda_1)$}
					\label{figA2}
				\end{figure}
			\end{enumerate}
			\newpage
			\begin{enumerate}
				\item[(iii)] $b_1,b_2 \in (\lambda_1,\lambda_2)$. The roots $\lambda_1$ and $\lambda_2$ lie in different intervals, in $(-\infty,b_1)$ and in $(b_2,a_1)$ respectively (see Figure \ref{figA3}).
				\begin{figure}[h!]
					\centering
					\includegraphics[scale=0.6]{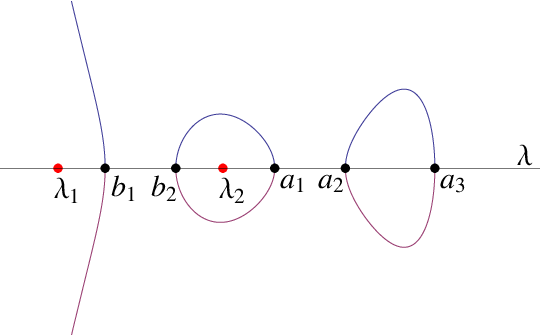}
					\caption{Real part of the spectral curve $\mu^{2}=f^{\mathbb C}(\lambda)$ in regard to the involution $\tau_{\{1,3\}}$ and case $\lambda_{1},\lambda_{2} \in (-\infty,a_1)$ and $b_1,b_2 \in (\lambda_1,\lambda_2)$}
					\label{figA3}
				\end{figure}
				\item[(iv)] $b_1,b_2 \in (\lambda_2,\infty)$. The roots $\lambda_1$ and $\lambda_2$ stay in the same interval, either in $(-\infty,b_1)$ or in $(-\infty,a_1)$ (see Figure \ref{figA4}, where $c_1,c_2,c_3,c_4$ and $c_5$ are ordered roots of polynomial $f(\lambda)$).
				\begin{figure}[h!]
					\centering
					\includegraphics[scale=0.6]{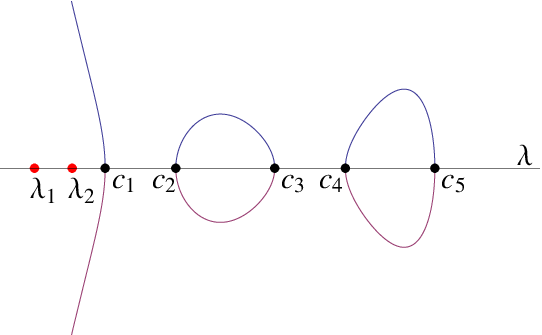}
					\caption{Real part of the spectral curve $\mu^{2}=f^{\mathbb C}(\lambda)$ in regard to the involution $\tau_{\{1,3\}}$ and case $\lambda_{1},\lambda_{2} \in (-\infty,a_1)$ and $b_1,b_2 \in (\lambda_2,\infty)$}
					\label{figA4}
				\end{figure}
			\end{enumerate}
			
			\noindent
			{\bf Case (b), $\lambda_{1},\lambda_{2} \in (a_1,a_2)$:} \ The roots $b_{1}$ and $b_{2}$ of the polynomial $f(\lambda)$ have to be real and $b_{1} \in (-\infty,\lambda_1)$ and $b_2 \in (\lambda_2,\infty)$. Thus the roots $\lambda_1$ and $\lambda_{2}$ lie in the common closed interval $(a_1,a_2)$ (see Figure \ref{fig01}, where $c_{1}$, $c_{2}$, $c_{3}$, $c_{4}$ and $c_{5}$ are ordered roots of polynomial $f(\lambda)$).
			\begin{figure}[h]
				\centering
				\includegraphics[scale=0.6]{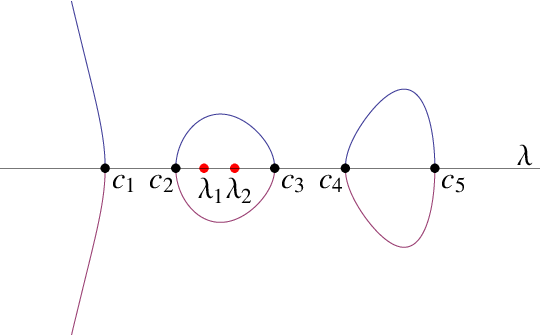}
				\caption{Real part of the spectral curve $\mu^{2}=f^{\mathbb C}(\lambda)$ in regard to the involution $\tau_{\{1,3\}}$ and case $\lambda_{1},\lambda_{2} \in (a_1,a_2)$}
				\label{fig01}
			\end{figure}

			\noindent
			{\bf Case (c), $\lambda_1,\lambda_2 \in (a_2,a_3)$:}  The possibilities of the positions of the roots $b_1$ and $b_2$ of the polynomial $f(\lambda)$ are the following
			\begin{enumerate}[(i)]
				\item $b_1$ and $b_2$ are complex roots and $f(\lambda)=-(\lambda-a_1)(\lambda-a_2)(\lambda-a_3)(\lambda^2+\beta^2)$ for some $\beta \in {\mathbb R}$. The real part of the spectral curve is the same as in Figure \ref{figA1}, only the roots $\lambda_1$ and $\lambda_2$ lie in $(a_2,a_3)$.
				\item $b_1,b_2 \in (-\infty,\lambda_1)$. The real part of the spectral curve is the same as in Figure \ref{figA4} ($c_1,c_2,c_3,c_4$ and $c_5$ are ordered roots of polynomial $f(\lambda)$ for this case), only the roots $\lambda_1$ and $\lambda_2$ lie in $(c_4,c_5)$.
				\item $b_1,b_2 \in (\lambda_1,\lambda_2)$. The roots $\lambda_1$ and $\lambda_2$ lie in different intervals, in $(a_2,b_1)$ and $(b_2,a_3)$ respectively (see Figure \ref{figC3}).
				\begin{figure}[h]
					\centering
					\includegraphics[scale=0.6]{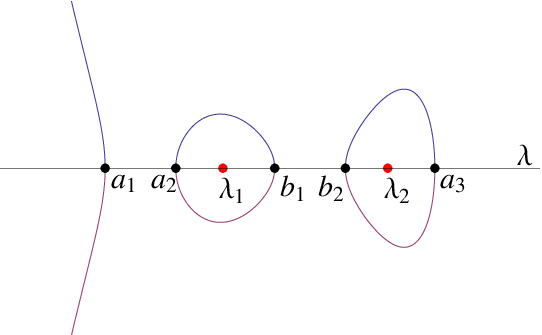}
					\caption{Real part of the spectral curve $\mu^{2}=f^{\mathbb C}(\lambda)$ in regard to the involution $\tau_{\{1,3\}}$ and case $\lambda_{1},\lambda_{2} \in (a_2,a_3)$ and $b_1,b_2 \in (\lambda_1,\lambda_2)$}
					\label{figC3}
				\end{figure}
				\item $b_1,b_2 \in (\lambda_2,\infty)$. The roots $\lambda_1$ and $\lambda_2$ stay in the same interval, either in $(a_2,a_3)$ or in $(a_2,b_1)$. The real part of the spectral curve is the same as in Figure \ref{fig01} ($c_1,c_2,c_3,c_4$ and $c_5$ are ordered roots of polynomial $f(\lambda)$ for this case) and $\lambda_1,\lambda_2 \in (c_2,c_3)$.
			\end{enumerate}
			
			\noindent
			{\bf Case (d), $\lambda_1,\lambda_2 \in (a_3,\infty)$:} The roots $b_{1}$ and $b_{2}$ of the polynomial $f(\lambda)$ have to be real and $b_{1} \in (-\infty,\lambda_{1})$ and $b_{2} \in (\lambda_{2},\infty)$. Thus the roots $\lambda_1$ and $\lambda_{2}$ lie in the common closed interval $(a_3,b_2)$ (see Figure \ref{fig02}, where $c_1,c_2,c_3,c_4$ and $c_5$ are ordered roots of polynomial $f(\lambda)$).
			\begin{figure}[h]
				\centering
				\includegraphics[scale=0.6]{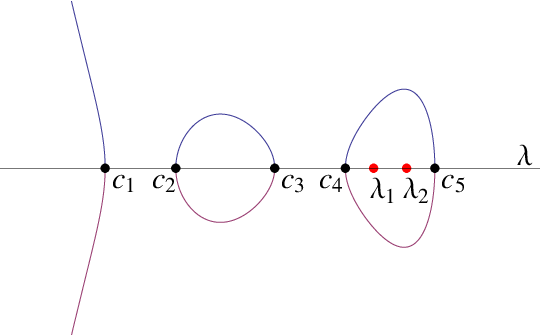}
				\caption{Real part of the spectral curve $\mu^{2}=f^{\mathbb C}(\lambda)$ in regard to the involution $\tau_{\{1,3\}}$ and case $\lambda_{1},\lambda_{2} \in (a_3,\infty)$ }
				\label{fig02}
			\end{figure}
			
			Taking into account the explanation of the topology of the level manifolds in \cite{Vanhaecke} (pages 86-89) the results above can be summarized. In the case (a)(iii) the Arnold-Liouville level set is isomorphic to a cylinder, in the case (c)(iii) the Arnold-Liouville level set is isomorphic to a torus. In all other cases the Arnold-Liouville level set is isomorphic to a disc with one hole.
			
			But there are still some outstanding issues. How does the position of roots determine a subset in ${\cal H}_{\cal S}^{2}$ and vice versa? Can we answer the questions asked in Introduction for some other case or in general? {\color{black} One can find the positive solution$(q_1^2,q_2^2,...,g_{n+1}^2)$ of the linear system and investigate correspondence between the position of roots $\lambda_1,\lambda_2,\ldots,\lambda_n$ and the separated part of the Hamiltonian system's phase space.} For a further study, the authors propose to consider the case $n=3$ first, to find out if there is any possibility to determine the signs of coefficients in a systems matrices. And perhaps to perceive whether the determination of the signs can be generalized to systems of arbitrary dimensions.

			\section*{Acknowledgments}
			
			The work was supported in part by ARRS, Research Agency of Slovenia (Grant number P2-0248).
			

			\appendix
			
			\section{Scripts and Functions in MATLAB}
			
			\label{appendix:scripts and fncs}
			
			
			Bellow, two scripts are presented. Scripts differ only in the fifth and sixth lines. In fifth line, a nonempty subset ${\cal S} \subseteq \{1,2,3\}$ is determined, and in sixth line, the positions of roots $\lambda_1$ and $\lambda_{2}$ (in MATLAB $l_1$ and $l_2$) are fixed. After calling a script the matrix System is the matrix of coefficients of the system (\ref{eq:systemformal}) for $n=2$ and the matrix SignsSytem is the matrix of signs of the coefficients in System, respectively.
			Recall that the notation $\rm{Snum_{1}Lnum_{2}}$ is described in subsection \ref{subsection:alg n2}. \\
			
			\noindent
			{\bf Script S1L00}
			
			\smallskip
			\begin{ttfamily}
				\noindent
				syms a1 a2 a3;\\
				syms l1 l2;\\
				a = [a1 a2 a3];\\
				l = [l1 l2];\\
				S = [1];\\
				assume(l1<l2<a1<a2<a3);\\
				n=2;\\
				A1=(l1-a1).*(l1-a2).*(l1-a3);\\
				A2=(l2-a1).*(l2-a2).*(l2-a3);\\
				A=[A1,A2];\\
				for r=1:2\\
				\hspace*{3mm} for s=1:3\\
				\hspace*{6mm} T(r,s) = (2.*ismember(s,S)-1).*A(r)/(l(r)-a(s));\\
				\hspace*{3mm} end;\\
				end;\\
				T0=[T, zeros(2,1)];\\
				System=[[2.*ismember([1 2 3],S)-1 -1];T0];\\
				SignsSystem=sign(System);\\
			\end{ttfamily}\\
			
			\noindent
			{\bf Script S13L12}
			
			\smallskip
			\begin{ttfamily}
				\noindent
				syms a1 a2 a3;\\
				syms l1 l2;\\
				a = [a1 a2 a3];\\
				l = [l1 l2];\\
				S = [1 3];\\
				assume(a1<l1<a2<l2<a3);\\
				n=2;\\
				A1=(l1-a1).*(l1-a2).*(l1-a3);\\
				A2=(l2-a1).*(l2-a2).*(l2-a3);\\
				A=[A1,A2];\\
				for r=1:2\\
				\hspace*{3mm} for s=1:3\\
				\hspace*{6mm} T(r,s) = (2.*ismember(s,S)-1).*A(r)/(l(r)-a(s));\\
				\hspace*{3mm} end;\\
				end;\\
				T0=[T, zeros(2,1)];\\
				System=[[2.*ismember([1 2 3],S)-1 -1];T0];\\
				SignsSystem=sign(System); \\
			\end{ttfamily}
			
			\noindent
			Two main functions of the algorithm are {\bf eops} and {\bf newsys}. Their functionality is explained in subsection \ref{subsection:alg n2}. \\
			
			\begin{ttfamily}
				\noindent
				function[D] = eops(B,n)\\
				i=0;\\
				while (i < n+1) \&\& (ismember(1,B(i+1,:))==1) \&\& (ismember(-1,B(i+1,:))==1)\\
				\hspace*{3mm}  i=i+1;\\
				end\\
				if i==n+1\\
				\hspace*{3mm}  D=[1 []];\\
				else\\
				\hspace*{3mm}  D=[0 B(i+1,:)];\\
				end\\
				end\\
			\end{ttfamily}\\
			
			\begin{ttfamily}
				\noindent
				function [Sys,B] = newsys(B11,B12,B21,B22) \\
				for j=1:2\\
				\hspace*{3mm}  A=((B11.')*B22(j,:)-((B21(j,:)).')*B12)';\\
				\hspace*{3mm}  Sys(j,:)=B(:)';\\
				end\\
				B=sign(factor(Sys));\\
				end\\
			\end{ttfamily}
			
			
			
			\newpage
			
			\section{Results of Algorithm}
			\label{appendix:results of alg}
			
			\noindent
			\begin{table}[h!]
				\caption{{\bf Yes/No (1/0) table for the existence of real roots of the systems for the subsets $\{1\}$, $\{2\}$ and $\{3\}$.}}
				\label{Tab:1}
				\begin{tabular}{|l|c|c||l|c|c||l|c|c|}
					\hline
					system roots & Y/N  & step & system roots & Y/N  & step & system roots & Y/N  & step \\ \hline
					S1L00        &  0   &  2   & S2L00        &  0   &  3   &   S3L00      &  0   &  3   \\ \hline
					S1L01        &  0   &  1   & S2L01        &  0   &  2   &   S3L01      &  1   &      \\ \hline
					S1L02        &  0   &  2   & S2L02        &  0   &  2   &   S3L02      &  0   &  1   \\ \hline
					S1L03        &  0   &  2   & S2L03        &  1   &      &   S3L03      &  0   &  2   \\ \hline
					S1L11        &  0   &  1   & S2L11        &  0   &  2   &   S3L11      &  0   &  3   \\ \hline
					S1L12        &  0   &  1   & S2L12        &  0   &  2   &   S3L12      &  0   &  1   \\ \hline
					S1L13        &  0   &  1   & S2L13        &  0   &  2   &   S3L13      &  0   &  2   \\ \hline
					S1L22        &  0   &  3   & S2L22        &  0   &  2   &   S3L22      &  0   &  1   \\ \hline
					S1L23        &  1   &      & S2L23        &  0   &  2   &   S3L23      &  0   &  1   \\ \hline
					S1L33        &  0   &  3   & S2L33        &  0   &  3   &   S3L33      &  0   &  2   \\ \hline
				\end{tabular}
			\end{table}
			\begin{table}[h!]
				\caption{{\bf Yes/No (1/0) table for the existence of real roots of the systems for the subsets $\{1,2\}$, $\{1,3\}$ and $\{2,3\}$.}}
				\label{Tab:2}
				\begin{tabular}{|l|c|c||l|c|c||l|c|c|}
					\hline
					system roots  & Y/N & step  & system roots & Y/N  & step  & system roots & R/NR & step  \\ \hline
					S12L00        &  0  &  2    & S13L00       &  1   &       &   S23L00     &  0   &  3pf  \\ \hline
					S12L01        &  0  &  2    & S13L01       &  0   &  3pf  &   S23L01     &  0   &  1    \\ \hline
					S12L02        &  0  &  1    & S13L02       &  0   &  3pf  &   S23L02     &  1   &       \\ \hline
					S12L03        &  0  &  2    & S13L03       &  0   &  3pf  &   S23L03     &  0   &  2    \\ \hline
					S12L11        &  0  &  3pf  & S13L11       &  1   &       &   S23L11     &  0   &  1    \\ \hline
					S12L12        &  0  &  1    & S13L12       &  0   &  3pf  &   S23L12     &  0   &  1    \\ \hline
					S12L13        &  1  &       & S13L13       &  0   &  3pf  &   S23L13     &  0   &  1    \\ \hline
					S12L22        &  0  &  1    & S13L22       &  1   &       &   S23L22     &  0   &  3pf  \\ \hline
					S12L23        &  0  &  1    & S13L23       &  0   &  3pf  &   S23L23     &  0   &  2    \\ \hline
					S12L33        &  0  &  3pf  & S13L33       &  1   &       &   S23L33     &  0   &  2    \\ \hline
				\end{tabular}
			\end{table}
			
			\bigskip
			
			\begin{table}[h!]
				\caption{{\bf Yes/No (1/0) table for the existence of real roots of the system for the subset $\{1,2,3\}$.}}
				\label{Tab:3}
				\begin{tabular}{|l|c|c|}
					\hline
					system roots  & Y/N & step  \\ \hline
					S123L00        &  0  &  1   \\ \hline
					S123L01        &  0  &  1   \\ \hline
					S123L02        &  0  &  1   \\ \hline
					S123L03        &  0  &  1   \\ \hline
					S123L11        &  0  &  3   \\ \hline
					S123L12        &  1  &      \\ \hline
					S123L13        &  0  &  1   \\ \hline
					S123L22        &  0  &  3   \\ \hline
					S123L23        &  0  &  1   \\ \hline
					S123L33        &  0  &  1   \\ \hline
				\end{tabular}
			\end{table}
			
			\bigskip
			\bigskip
			
			\newpage
			
			\section{Additional analysis of the results in Table \ref{Tab:2}}
			\label{appendix:analysis}

			\noindent
			{\bf Example S12L11:}
			In the third step of the algorithm, the matrix of signs (i.e. the matrix $B$) is
			$$\left[\begin{matrix}
			-1 \\ -\sgn(a_1a_2-a_1a_3+a_2a_3-a_2\lambda_1-a_2\lambda_2+\lambda_1\lambda_2) \\ \sgn(a_1a_2+a_1a_3-a_2a_3-a_1\lambda_1-a_1\lambda_2+\lambda_1\lambda_2) \\ -1
			\end{matrix}\right]^{T}.$$
			The expressions $a_1a_2-a_1a_3+a_2a_3-a_2\lambda_1-a_2\lambda_2+\lambda_1\lambda_2$ and $a_1a_2+a_1a_3-a_2a_3-a_1\lambda_1-a_1\lambda_2+\lambda_1\lambda_2$ can be partially factorized as
			$$ a_1a_2-a_1a_3+a_2a_3-a_2\lambda_1-a_2\lambda_2+\lambda_1\lambda_2 = (\lambda_1-a_1)(a_3-a_2)+(a_3-\lambda_2)(a_2-\lambda_1) $$
			and
			$$ a_1a_2+a_1a_3-a_2a_3-a_1\lambda_1-a_1\lambda_2+\lambda_1\lambda_2 = (a_2-\lambda_1)(a_1-a_3)+(a_1-\lambda_1)(a_3-\lambda_2). $$
			Since $a_1<\lambda_1<\lambda_2<a_2<a_3$, the first expression is the sum of two positive summands and the second expression is the sum of two negative summands. Thus, the matrix of signs is $[-1, -1, -1, -1]$ and a positive solution does not exist.

			\vspace{5mm}
			\noindent
			{\bf Example S12L33:}
			In the third step of the algorithm, the matrix of signs is
			$$\left[\begin{matrix}
			1 \\ \sgn(a_1a_2+a_1a_3-a_2a_3-a_1\lambda_1-a_1\lambda_2+\lambda_1\lambda_2) \\ \sgn(a_1a_2-a_1a_3+a_2a_3-a_2\lambda_1-a_2\lambda_2+\lambda_1\lambda_2) \\ 1
			\end{matrix}\right]^{T}.$$
			As in Example S12L11, the expressions can be partially factorized as
			$$ a_1a_2+a_1a_3-a_2a_3-a_1\lambda_1-a_1\lambda_2+\lambda_1\lambda_2 = (a_2-\lambda_1)(a_1-a_3)+(a_1-\lambda_1)(a_3-\lambda_2) $$
			and
			$$ a_1a_2-a_1a_3+a_2a_3-a_2\lambda_1-a_2\lambda_2+\lambda_1\lambda_2 = (\lambda_1-a_1)(a_3-a_2)+(a_3-\lambda_{2})(a_2-\lambda_1). $$
			Since $a_1<a_2<a_3<\lambda_1<\lambda_2$, the first expression is the sum of two positive summands (both are the product of two negative factors) and the second expression is also the sum of two positive summands. Therefore, the matrix of signs is $[1,1,1,1]$ and a positive solution does not exist.

			\vspace{5mm}
			\noindent
			{\bf Example S13L01:}
			In the third step of the algorithm, the matrix of signs is
			$$\left[\begin{matrix}
			1, & -\sgn(a_1a_2+a_1a_3-a_2a_3-a_1\lambda_1-a_1\lambda_2+\lambda_1\lambda_2), & 1
			\end{matrix}\right].$$
			If we use the partial factorization in Example S12L11
			$$ a_1a_2+a_1a_3-a_2a_3-a_1\lambda_1-a_1\lambda_2+\lambda_1\lambda_2 = (a_2-\lambda_1)(a_1-a_3)+(a_1-\lambda_1)(a_3-\lambda_2) ,$$
			from the ordering $\lambda_1 < a_1 < \lambda_2 < a_2 < a_3$ we can not deduce on the sign of the expression. But the expression can be also partially factorized as
			$$ a_1a_2+a_1a_3-a_2a_3-a_1\lambda_1-a_1\lambda_2+\lambda_1\lambda_2 = (a_2-\lambda_1)(a_1-\lambda_2) + (a_3-\lambda_2)(a_1-a_2) ,$$
			and the negativity of the expression is obvious. The matrix of signs is $[1, 1, 1]$ and a positive solution does not exist.

			\vspace{5mm}
			\noindent
			{\bf Example S13L02:}
			In the third step of the algorithm, the matrix of signs is
			$$\left[\begin{matrix}
			1, & -\sgn(a_1a_2+a_1a_3-a_2a_3-a_1\lambda_1-a_1\lambda_2+\lambda_1\lambda_2), & 1
			\end{matrix}\right].$$
			As in the previous example, the expression $a_1a_2+a_1a_3-a_2a_3-a_1\lambda_1-a_1\lambda_2+\lambda_1\lambda_2$ can be partially factorized as
			$$ a_1a_2+a_1a_3-a_2a_3-a_1\lambda_1-a_1\lambda_2+\lambda_1\lambda_2 = (a_2-\lambda_1)(a_1-\lambda_2)+(a_3-\lambda_2)(a_1-a_2).$$
			Since $\lambda_1 < a_1 < a_2 < \lambda_2 < a_3$, the expression is the sum of two negative summands. The matrix of signs is $[1,1,1]$ and a positive solution does not exist.

			\vspace{5mm}
			\noindent
			{\bf Example S13L03:}
			In the third step of the algorithm, the matrix of signs is
			$$\left[\begin{matrix}
			1, & -\sgn(a_1a_2+a_1a_3-a_2a_3-a_1\lambda_1-a_1\lambda_2+\lambda_1\lambda_2), & 1
			\end{matrix}\right].$$
			Since
			$$ a_1a_2+a_1a_3-a_2a_3-a_1\lambda_1-a_1\lambda_2+\lambda_1\lambda_2 = (a_2-\lambda_1)(a_1-\lambda_2)+(a_3-\lambda_2)(a_1-a_2)$$
			(see previous two examples) and $\lambda_1 < a_1 < a_2 < a_3 < \lambda_2$, the matrix of signs is $[1,1,1]$ and a positive solution does not exist.

			\vspace{5mm}
			\noindent
			{\bf Example S13L12:}
			In the third step of the algorithm, the matrix of signs is
			$$\left[\begin{matrix}
			\sgn(a_1a_2-a_1a_3-a_2a_3+a_3\lambda_1+a_3\lambda_2-\lambda_1\lambda_2), &  1, & 1
			\end{matrix}\right].$$
			The expression $a_1a_2-a_1a_3-a_2a_3+a_3\lambda_1+a_3\lambda_2-\lambda_1\lambda_2$ can be partially factorized as
			$$ a_1a_2-a_1a_3-a_2a_3+a_3\lambda_1+a_3\lambda_2-\lambda_1\lambda_2 = (\lambda_1-a_1)(a_3-a_2)+(\lambda_1-a_3)(a_2-\lambda_2). $$
			Since $a_1 < \lambda_1 < a_2 < \lambda_2 < a_3$, the considered expression is the sum of two positive summands. The matrix of signs is $[1,1,1]$ and a positive solution does not exist.

			\vspace{5mm}
			\noindent
			{\bf Example S13L13:}
			In the third step of the algorithm, the matrix of signs is
			$$\left[\begin{matrix}
			\sgn(a_1a_2-a_1a_3-a_2a_3+a_3\lambda_1+a_3\lambda_2-\lambda_1\lambda_2), &  1, & 1
			\end{matrix}\right].$$
			Since
			$$ a_1a_2-a_1a_3-a_2a_3+a_3\lambda_1+a_3\lambda_2-\lambda_1\lambda_2 = (\lambda_1-a_1)(a_3-a_2)+(\lambda_1-a_3)(a_2-\lambda_2) $$
			(see the previous example) and $a_1 < \lambda_1 < a_2 < a_3 < \lambda_2$, the considered expression is the sum of two positive summands. The matrix of signs is $[1,1,1]$ and a positive solution does not exist.

			\vspace{5mm}
			\noindent
			{\bf Example S13L23:}
			In the third step of the algorithm, the matrix of signs is
			$$\left[\begin{matrix}
			1, & \sgn(a_1a_2+a_1a_3-a_2a_3-a_1\lambda_1-a_1\lambda_2+\lambda_1\lambda_2), &  1
			\end{matrix}\right].$$
			Since
			$$ a_1a_2+a_1a_3-a_2a_3-a_1\lambda_1-a_1\lambda_2+\lambda_1\lambda_2 = (a_2-\lambda_1)(a_1-a_3)+(a_1-\lambda_1)(a_3-\lambda_2) $$
			(see Example S12L11) and $a_1 < a_2 < \lambda_1 < a_3 < \lambda_2$, the considered expression is the sum of two positive summands. The matrix of signs is $[1,1,1]$ and a positive solution does not exist.

			\vspace{5mm}
			\noindent
			{\bf Example S23L00:}
			In the third step of the algorithm, the matrix of signs is
			$$\left[\begin{matrix}
			-1 \\
			-\sgn(a_1a_2-a_1a_3+a_2a_3-a_2\lambda_1-a_2\lambda_2+\lambda_1\lambda_2) \\
			\sgn(a_1a_2-a_1a_3-a_2a_3+a_3\lambda_1+a_3\lambda_2-\lambda_1\lambda_2) \\
			-1
			\end{matrix}\right]^{T}. $$
			If we use the partial factorization in Example S12L11
			$$ a_1a_2-a_1a_3+a_2a_3-a_2\lambda_1-a_2\lambda_2+\lambda_1\lambda_2 = (\lambda_1-a_1)(a_3-a_2) + (a_3-\lambda_2)(a_2-\lambda_1), $$
			from the ordering $\lambda_1 < \lambda_2 < a_1 < a_2 < a_3$ we can not deduce on the sign of the expression. After rearranging
			$$ a_1a_2-a_1a_3+a_2a_3-a_2\lambda_1-a_2\lambda_2+\lambda_1\lambda_2 = (\lambda_1-a_1)(\lambda_2-a_2) + (a_3-\lambda_2)(a_2-a_1), $$
			the positivity of the expression obvious. Since
			$$ a_1a_2-a_1a_3-a_2a_3+a_3\lambda_1+a_3\lambda_2-\lambda_1\lambda_2 =(\lambda_1-a_1)(a_3-a_2)+(\lambda_1-a_3)(a_2-\lambda_2) $$
			(see Example S13L12), the expression is the sum of two negative summands. Therefore, the matrix of signs is $[-1,-1,-1,-1]$ and a positive solution does not exist.

			\vspace{5mm}
			\noindent
			{\bf Example S23L22:}
			In the third step of the algorithm, the matrix of signs is
			$$\left[\begin{matrix}
			1 \\
			\sgn(a_1a_2-a_1a_3+a_2a_3-a_2\lambda_1-a_2\lambda_2+\lambda_1\lambda_2) \\
			\sgn(a_1a_2-a_1a_3-a_2a_3+a_3\lambda_1+a_3\lambda_2-\lambda_1\lambda_2) \\
			1
			\end{matrix}\right]^{T}. $$
			Since
			$$ a_1a_2-a_1a_3+a_2a_3-a_2\lambda_1-a_2\lambda_2+\lambda_1\lambda_2 = (\lambda_1-a_1)(\lambda_2-a_2)+(a_3-\lambda_2)(a_2-a_1) $$
			and
			$$ a_1a_2-a_1a_3-a_2a_3+a_3\lambda_1+a_3\lambda_2-\lambda_1\lambda_2 = (\lambda_1-a_1)(a_3-a_2)+(\lambda_1-a_3)(a_2-\lambda_2), $$
			(see the previous example) and $a_1 < a_2 < \lambda_1 < \lambda_2 < a_3$, both expressions are the sums of two positive summands. The matrix of signs is $[1,1,1,1]$ and a positive solution does not exist.

			
			\bigskip
			
			
			
			\noindent
			{\bf References}
			
			
			
			

		\end{document}